\documentclass[conference]{IEEEtran}

\usepackage[T1]{fontenc}
\usepackage{amsmath}
\usepackage{txfonts}

\newcommand{\qed}{\hspace*{\fill}\rule{1ex}{1ex}}

\def\Label#1{\label{#1}\ [\ #1\ ]\ }
\def\Label{\label}

\title{Universally Attainable Error and Information Exponents, and Equivocation Rate
for the Broadcast Channels with Confidential Messages}
\author{\IEEEauthorblockN{Masahito Hayashi}\IEEEauthorblockA{%
Graduate School of Information Sciences,\\
Tohoku University, 980-8579 Japan\\
and
Centre for Quantum Technologies,\\
 National University of Singapore,\\
 3 Science Drive 2, Singapore 117542}
 \and \IEEEauthorblockN{Ryutaroh Matsumoto}\IEEEauthorblockA{%
Department of Communications and Integrated Systems,\\
Tokyo Instiutte of Technology, 152-8550 Japan}}
\date{\today}

\newtheorem{definition}{Definition}
\newtheorem{proposition}[definition]{Proposition}
\newtheorem{theorem}[definition]{Theorem}
\newtheorem{corollary}[definition]{Corollary}
\newtheorem{remark}[definition]{Remark}

\allowdisplaybreaks
\begin{document}
\maketitle
\thispagestyle{plain}
\begin{abstract}
We show universally attainable exponents 
for the decoding error and the mutual information 
and universally attainable equivocation rates for the conditional 
entropy for the
broadcast channels with confidential messages.
The error exponents are the same as ones given by K\"orner
and Sgarro for the broadcast channels with degraded message sets.
\end{abstract}

\section{Introduction}\Label{sec1}
The information theoretic security attracts much attention
recently \cite{liang09},
because it offers security that does not depend on
a conjectured difficulty of some computational problem.
A classical problem in the information theoretic security
is the broadcast channel with confidential messages
(hereafter abbreviated as BCC)
first considered by Csisz\'ar and K\"orner \cite{csiszar78},
in which there is a single sender called Alice and
two receivers called Bob and Eve.
The problem in \cite{csiszar78} is a generalization of
the wiretap channel considered by Wyner \cite{wyner75}.
In the formulation in \cite{csiszar78},
Alice has a common messages destined for both Bob and Eve and
a private message destined solely for Bob.
The word ``confidential'' means that Alice wants to
prevent Eve from knowing much about the private message.
The coding in this situation has two goals, namely
error correction and secrecy.
The degree of secrecy is measured by the mutual information
between the private message to Bob and the received signal
by Eve.

On the other hand, the broadcast channel with degraded
message sets (hereafter abbreviated as BCD), considered by K\"orner and Marton \cite{korner77}, is a special case of BCC in which there is no requirement on the
confidentiality of the private message to Bob.
K\"orner and Sgarro \cite{korner80}
proposed the universal encoder and decoder
for BCD, which did not use conditional probability distribution of the channel
for encoding nor decoding, and clarified universally attainable
error exponents for BCD.
In studies of universal coding for channels,
we consider the compound channel, which is a
collection of (usually infinitely many) channels and try to clarify the exponents
realized by given encoder and decoder over that collection of channels.
The capacity of the compound wiretap channel was studied by
Liang et al.\ \cite{liang09wiretap},
in which the number of channels is finite and the receiver is
assumed to know the conditional probability distribution of the channel.
Kobayashi et al.\ \cite{kobayashi09} studied the BCC under the
same assumption as \cite{liang09wiretap}.
The assumption in \cite{kobayashi09,liang09wiretap} is more restrictive
than \cite{korner80} and they \cite{kobayashi09,liang09wiretap}
only proved that the mutual information divided by the code length
converges to zero, which means that their universally attainable
information exponents are zero.

In contrast to them,
Soma \cite{soma10}
clarified the universally attainable 
information exponents for the wiretap channels
under the same assumption as \cite{korner80}, though
he did not analyze the universally attainable error exponent
by his coding scheme.
Soma \cite{soma10} used  the channel resolvability
lemma in \cite{hayashi06b}, and it was totally unclear
how to extend his argument to the BCC with common messages.
Soma did not clarified the speed of convergence of mutual
information to the infinity when the information rate
of private message is large, neither.
We note that our universally attainable information exponent is
the same as Soma's \cite{soma10} when there is no common message.

In this paper, we attach the two-universal hash functions
\cite{carter79} to the encoder proposed by K\"orner and
Sgarro \cite{korner80}, then
we use the privacy amplification theorem \cite{bennett95privacy}
for the analysis of the mutual information
to obtain the universally attainable error and information exponents 
and universally attainable equivocation rates for the BCC.
Our argument is similar to the non-universal coding considered
in \cite{matsumotohayashi2011ieice} and the secure multiplex coding
considered in \cite{matsumotohayashi2011isit}.

This paper is organized as follows:
Section \ref{sec2} reviews relevant research results
used in this paper.
Section \ref{sec3} introduces the definition of
universally attainable exponents and provides ones satisfying the
definition.
Section \ref{sec4} concludes the paper.

\section{Preliminary}\Label{sec2}
\subsection{Broadcast channels with confidential messages}
Let Alice, Bob, and Eve be as defined in Section \ref{sec1}.
$\mathcal{X}$ denotes the channel input alphabet and
$\mathcal{Y}$ (resp.\ $\mathcal{Z}$) denotes
the channel output alphabet to Bob (resp.\ Eve).
We assume that $\mathcal{X}$, $\mathcal{Y}$, and
$\mathcal{Z}$ are finite unless otherwise stated.
%We shall discuss the continuous channel briefly in Remark
%\ref{rem:cont}.
We denote the conditional probability of the channel
to Bob (resp.\ Eve) by $P_{Y|X}$ (resp.\ $P_{Z|X}$).
The set $\mathcal{S}_n$ denotes that of the secret message
and $\mathcal{E}_n$ does that of the common message
when the block coding of length $n$ is used.
We shall define the achievability of a rate triple
$(R_{\mathrm{s}}$, $R_e$, $R_{\mathrm{c}})$.
For the notational convenience, we fix the base of logarithm,
including one used in entropy and mutual information, to the
base of natural logarithm.
The privacy amplification theorem introduced in
Theorem \ref{thm:pa}
is sensitive to
choice of the base of logarithm.

\begin{definition}
The rate triple $(R_{\mathrm{s}}$, $R_e$, $R_{\mathrm{c}})$ is said
to be \emph{achievable} if there exists a sequence of
Alice's stochastic encoder $f_n$ from 
$\mathcal{S}_n \times \mathcal{E}_n$ to $\mathcal{X}^n$,
Bob's deterministic decoder $\varphi_n: \mathcal{Y}^n
\rightarrow \mathcal{S}_n \times \mathcal{E}_n$ and
Eve's deterministic decoder $\psi_n: \mathcal{Z}^n
\rightarrow \mathcal{E}_n$ such that
\begin{eqnarray*}
\lim_{n\rightarrow \infty} \mathrm{Pr}[
(S_n, E_n) \neq \varphi_n(Y^n) \textrm{ or } E_n \neq \psi_n(Z^n)]
&=& 0,\\
\liminf_{n\rightarrow \infty} \frac{H(S_n|Z^n)}{n} &\geq& R_e,\\
\liminf_{n\rightarrow\infty} \frac{\log|\mathcal{S}_n|}{n} &\geq& R_{\mathrm{s}},\\
\liminf_{n\rightarrow\infty} \frac{\log|\mathcal{E}_n|}{n} &\geq& R_{\mathrm{c}},
\end{eqnarray*}
where $S_n$ and $E_n$ represents the secret and the common message,
respectively, have the uniform distribution on $\mathcal{S}_n$
and $\mathcal{E}_n$, respectively,
and $Y^n$ and $Z^n$ are the received signal by Bob and Eve,
respectively, with the transmitted signal $f_n(S_n,E_n)$
and the channel transition probabilities $P_{Y|X}$, $P_{Z|X}$.
The capacity region of the BCC is the closure of 
the achievable rate triples.
\end{definition}

\begin{theorem}\Label{th1}\cite{csiszar78}
The capacity region for the BCC is given by
the set of $R_{\mathrm{c}}$, $R_{\mathrm{s}}$ and $R_e$ such that
there exists a Markov chain $U\rightarrow V \rightarrow X \rightarrow 
YZ$ and
\begin{eqnarray}
R_{\mathrm{s}} + R_{\mathrm{c}} &\leq& I(V;Y|U)+\min[I(U;Y),I(U;Z)],\Label{eq:bcc1}\\
R_{\mathrm{c}} &\leq& \min[I(U;Y),I(U;Z)],\Label{eq:bcc2}\\
R_e & \leq & I(V;Y|U)-I(V;Z|U),\nonumber\\
R_e & \leq &R_{\mathrm{s}}.\nonumber
\end{eqnarray}
\end{theorem}

As described in \cite{liang09},
$U$ can be regarded as the common message,
$V$ the combination of the common and the private messages,
and $X$ the transmitted signal.

%
%If we set $R_e = R_{\mathrm{s}}$ then we have
%$\lim_{n\rightarrow\infty} I(S_n;Z^n)/n =0$,
%which is traditionally called \emph{perfect security},
%because Eve knows little about $S_n$.
%However, Maurer \cite{maurer94} and Csisz\'ar \cite{csiszar96}
%observed that 
%$\lim_{n\rightarrow\infty} I(S_n;Z^n) =0$ is a better
%criterion for the secrecy of $S_n$ from Eve,
%and this stronger requirement is called the strong security criterion,
%while the traditional one is called the weak security criterion
%recently.
%
\begin{corollary}\Label{cor:bcc}\cite{csiszar78}
The notation is same as Theorem \ref{th1}.
If we require $R_e = R_{\mathrm{s}}$, the capacity region for
$(R_{\mathrm{c}}$, $R_{\mathrm{s}})$ is given by
the set of $R_{\mathrm{c}}$ and $R_{\mathrm{s}}$ such that
there exists a Markov chain $U\rightarrow V \rightarrow X \rightarrow 
YZ$ and
\begin{eqnarray*}
R_{\mathrm{c}} &\leq& \min[I(U;Y),I(U;Z)],\\
R_{\mathrm{s}} & \leq & I(V;Y|U)-I(V;Z|U).
\end{eqnarray*}
\end{corollary}

\subsection{Broadcast channels with degraded message sets}\Label{sec:bcd}
If we set $R_e= 0$ in the BCC,
the secrecy requirement is removed from BCC, and
the coding problem is equivalent to
the broadcast channel with degraded message sets (abbreviated as
BCD) considered by K\"orner and Marton \cite{korner77}.
\begin{corollary}\Label{cor:bcd}
The capacity region of the BCD is given by
the set of $R_{\mathrm{c}}$ and $R_{\mathrm{p}}$ such that
there exists a Markov chain $U\rightarrow V = X \rightarrow
YZ$ and
\begin{eqnarray}
R_{\mathrm{c}} &\leq& \min[I(U;Y),I(U;Z)],\nonumber\\
R_{\mathrm{c}}+R_{\mathrm{p}} & \leq & I(V;Y|U)+\min[I(U;Y),I(U;Z)]\Label{eq:bcd1},
\end{eqnarray}
where $R_{\mathrm{p}}$ denotes the rate of the private message.
\end{corollary}

\subsection{Two-universal hash functions}
We shall use a family of two-universal hash functions \cite{carter79}
for the privacy amplification theorem introduced later.

\begin{definition}
Let $\mathcal{F}$ be a set of functions from $\mathcal{S}_1$ to $\mathcal{S}_2$,
and $F$ the not necessarily uniform
random variable on $\mathcal{F}$. If for any $x_1 \neq x_2
\in \mathcal{S}_1$ we have
\[
\mathrm{Pr}[F(x_1)=F(x_2)] \leq \frac{1}{|\mathcal{S}_2|},
\]
then $\mathcal{F}$ is said to be a \emph{family of two-universal hash functions}.
\end{definition}
%We shall use a family of two-universal hash functions in which each function is surjective.
%The set of all surjective linear functions is an example of such a family.

\subsection{Strengthened privacy amplification theorem}
In order to analyze the equivocation rate, we need to strengthen the privacy amplification theorem
originally appeared in \cite{bennett95privacy,hayashi11}.
\begin{theorem}\Label{thm:pa} \cite{matsumotohayashi2011isit,matsumotohayashi2011netcod}
Let $L$ be a random variable with a finite alphabet
$\mathcal{L}$ and $Z$ any random variable. %\footnote{We do not
%assume the existence of its probability mass function nor
%probability density function.}.
Let $\mathcal{F}$ be a family of two-universal
hash functions from $\mathcal{L}$ to $\mathcal{M}$,
and $F$ be a random variable on $\mathcal{F}$
statistically independent of $L$.
Then
\begin{equation}
\mathbf{E}_f \exp(\rho I(F(L);Z|F=f))  \leq   
1+ |\mathcal{M}|^\rho\mathbf{E}[P_{L|Z}(L|Z)^\rho]\Label{hpa1}
\end{equation}
for $0<\rho\leq 1$.
%If $Z$ is not discrete RV,
%$I(F(L);Z|F)$ is defined to be
%$H(F(L)|F) - \mathbf{E}_z H(F(L)|F,Z=z)$.

In addition to the above assumptions,
when $L$ is uniformly distributed, we have
\begin{equation}
|\mathcal{M}|^\rho\mathbf{E}[P_{L|Z}(L|Z)^\rho]=
\frac{|\mathcal{M}|^\rho\mathbf{E}[P_{L|Z}(L|Z)^\rho P_L(L)^{-\rho}]}{|\mathcal{L}|^\rho}.\Label{hpa1uni}
\end{equation}
In addition to all of the above assumptions,
when $Z$ is a discrete random variable, 
we have
\begin{equation}
\frac{|\mathcal{M}|^\rho\mathbf{E}[P_{L|Z}(L|Z)^\rho P_L(L)^{-\rho}]}{|\mathcal{L}|^\rho}=
\frac{|\mathcal{M}|^\rho}{|\mathcal{L}|^\rho}
\sum_{z,\ell} P_L(\ell) P_{Z|L}(z|\ell)^{1+\rho} P_Z(z)^{-\rho}.
\Label{hpa1discrete}
\end{equation}
\end{theorem}

As in \cite{hayashi11} we introduce the following two functions.
\begin{definition}
\begin{eqnarray}
\psi(\rho, P_{Z|L}, P_L) &=& 
\log \sum_z \sum_\ell P_L(\ell) P_{Z|L}(z|\ell)^{1+\rho} P_Z(z)^{-\rho},
\Label{eq:psid}\\
\phi(\rho,P_{Z|L},P_L) 
&=& \log \sum_z\left(
\sum_{\ell} P_{L}(\ell) (P_{Z|L}(z|\ell)^{1/(1-\rho)})\right)^{1-\rho}.
\Label{phid}
\end{eqnarray}
\end{definition}
Observe that $\phi$ is essentially Gallager's function $E_0$
\cite{gallager68}.

At the end of our evaluation of the mutual information to Eve,
we shall use the averaged version of $\phi$, which is introduced
below.
\begin{definition}
\begin{eqnarray}
&&\phi(\rho,P_{Z|L},P_{L|U}, P_U) \nonumber\\
&=&  \log \sum_u P_U(u)\sum_z \left(
\sum_{\ell} P_{L|U}(\ell|u) (P_{Z|L}(z|\ell)^{1/(1-\rho)})\right)^{1-\rho}.
\Label{phidu}
\end{eqnarray}
\end{definition}

\begin{proposition}
For fixed $0<\rho \le  1$, $P_L$, $\tilde{P}_L$, $P_{Z|L}$, and $\tilde{P}_{Z|L}$ we have
\begin{align}
\exp(\psi(\rho, P_{Z|L}, P_L))
& \leq \exp(\phi(\rho, P_{Z|L}, P_L)).\Label{eq:philarger} \\
\exp(\phi(\rho, P_{Z|L}, P_L))
& \le 
C_1 \exp(\phi(\rho, P_{Z|L}, \tilde{P}_L)) \Label{Haya-2} 
\end{align}
when $
P_L \le C_1 \tilde{P}_L$.
\end{proposition}

\begin{proof}
The first inequality (\ref{eq:philarger}) was shown in \cite{hayashi11}.

Any positive concave function $f$ of probability
distributions satisfies 
\begin{equation}
f(P) \leq \alpha f(Q), \Label{eq:concave}
\end{equation}
when $P \leq \alpha \times Q$ with a positive real number $\alpha \ge 1$.
This is 
because by the assumption there exists
another distribution $R$ such that $(1/\alpha) P + (\alpha-1)/\alpha \cdot R
= Q$, and
\begin{eqnarray*}
f(P)/\alpha & \leq & f(P)/\alpha + (\alpha-1)/\alpha \cdot f(R)\\
&\leq & f((1/\alpha) P + (\alpha-1)/\alpha \cdot R) = f(Q).
\end{eqnarray*}
Since $\exp(\phi(\rho, P_{Z|L}, P_L))$ is concave with respect to $P_L$ with fixed $0<\rho< 1$ and $P_{Z|L}$ \cite{gallager68}, 
the second inequality (\ref{Haya-2}) holds.
\end{proof}

\section{Universal coding for the broadcast channels with
confidential messages}\Label{sec3}
\subsection{Universally attainable exponents and universally attainable equivocation rates}
We introduce the universally attainable exponents for the BCC
by adjusting the original definition for the BCD given by
K\"orner and Sgarro \cite{korner80}.
\begin{definition}\Label{def:univexp}
Let $\mathcal{W}(\mathcal{X}$, $\mathcal{Y}$, $\mathcal{Z})$
be the set of all discrete memoryless broadcast channels
$W: \mathcal{X} \rightarrow \mathcal{Y}, \mathcal{Z}$, and $\mathbf{R}^+$ the set of positive
real numbers.
A quadruple of functions $(\tilde{E}^{\mathrm{p}}$, $\tilde{E}^{\mathrm{c}}$,
$\tilde{E}_+^{\mathrm{I}}$, 
$\tilde{E}_-^{\mathrm{I}})$ 
from $\mathbf{R}^+ \times \mathbf{R}^+ 
\times \mathcal{W}(\mathcal{X}$, $\mathcal{Y}$, $\mathcal{Z})$
to $[\mathbf{R}^+ \cup \{0\}]^4$ is said to be
a universally attainable quadruple of exponents and equivocation rate
for the
family $\mathcal{W}(\mathcal{X}$, $\mathcal{Y}$, $\mathcal{Z})$
if, for every $R_{\mathrm{s}} > 0$, $R_{\mathrm{c}}>0$, $\delta>0$,
and for sufficiently large $n$,
there exists a sequence of codes $(f_n$, $\varphi_n$, $\psi_n)$
of length $n$ of rate pair at least
$(R_{\mathrm{s}}$, $R_{\mathrm{c}})$ such that,
denoting by $e_n^{\mathrm{s}}(W)$, $e_n^{\mathrm{c}}(W)$,
$e_n^{\mathrm{I}}(W)$ the maximum error probabilities by Bob and by Eve
and the mutual information between the secret message and Eve's received
signal, for the $n$-th memoryless extension of the channel $W
\in \mathcal{W}(\mathcal{X}$, $\mathcal{Y}$, $\mathcal{Z})$
we have
\begin{align}
e_n^{\mathrm{s}}(W) \leq & \exp(-n
[ \tilde{E}^{\mathrm{s}}(R_{\mathrm{s}}, R_{\mathrm{c}})-\delta]),\Label{Haya-51}\\
e_n^{\mathrm{c}}(W) \leq & \exp(-n
[ \tilde{E}^{\mathrm{c}}(R_{\mathrm{s}},  R_{\mathrm{c}})-\delta]),\Label{Haya-52}\\
e_n^{\mathrm{I}}(W) \leq & 
\max \{
\exp(-n
[ \tilde{E}_+^{\mathrm{I}}(R_{\mathrm{s}}, R_{\mathrm{c}})-\delta]),
n [\tilde{E}_-^{\mathrm{I}}(R_{\mathrm{s}}, R_{\mathrm{c}})+\delta ]
\}
,\Label{eq:logexp}
\end{align}
where $R_{\mathrm{s}}$ and $R_{\mathrm{c}}$ denote
the rate of the secret message and the common message, respectively.
\end{definition}

Suppose that we are given a broadcast $W:\mathcal{X}\rightarrow \mathcal{Y},\mathcal{Z}$ and positive real numbers
$R_{\mathrm{s}}$ and $R_{\mathrm{c}}$.
We fix a distribution $Q_{UV}$ on $\mathcal{U}\times \mathcal{V}$,
a channel $\Xi: \mathcal{V}\rightarrow \mathcal{X}$, and the rate
$R_{\mathrm{p}}$ of the private message in the BCD encoder
that satisfy Eqs.\ (\ref{eq:bcc1}), (\ref{eq:bcc2}) and (\ref{eq:bcd1}),
where the RVs $U$ ,$V$, $X$, $Y$ and $Z$ in Eqs.\ (\ref{eq:bcc1}), (\ref{eq:bcc2}) and (\ref{eq:bcd1})
are distributed according to $Q_{UV}$,
$\Xi$ and $W$.
We present a universally attainable quadruple of exponents and equivocation rate
in terms of $R_{\mathrm{p}}$, $Q_{UV}$ and $\Xi$ as
\begin{align}
F^{\mathrm{s}} =& F^{\mathrm{s}}(W, R_{\mathrm{p}}, R_{\mathrm{c}}, Q_{UV},\Xi ) = 
F^{\mathcal{Y},\mathrm{KS}}(W\circ \Xi, R_{\mathrm{p}}, R_{\mathrm{c}}, Q_{UV}),
\Label{eq:univs}\\
F^{\mathrm{c}} =& F^{\mathrm{c}}(W, R_{\mathrm{p}}, R_{\mathrm{c}}, Q_{UV},\Xi ) = 
F^{\mathcal{Z},\mathrm{KS}}(W\circ \Xi, R_{\mathrm{p}}, R_{\mathrm{c}}, Q_{UV}),
\Label{eq:univc}\\
F_+^{\mathcal{I}} 
=& F_+^{\mathcal{I}}(W, R_{\mathrm{p}}, R_{\mathrm{s}}, Q_{UV},\Xi ) 
\nonumber \\
=&
\sup_{0 < \rho \le 1} \Biggl[\rho (R_{\mathrm{p}}-R_{\mathrm{s}}) -
\phi (\rho,W_{\mathcal{Z}}\circ \Xi ,Q_{V|U}, Q_{U} )
\Biggr] 
\Label{eq:univi1}\\
F_-^{\mathcal{I}} 
=& F_-^{\mathcal{I}}(W, R_{\mathrm{p}}, R_{\mathrm{s}}, Q_{UV},\Xi ) 
= 
I(V;Z|U) -R_{\mathrm{p}}+R_{\mathrm{s}} 
\Label{eq:univi2},
\end{align}
where
$F^{\mathcal{Y},\mathrm{KS}}$ and $F^{\mathcal{Z},\mathrm{KS}}$ are the error exponent functions appeared in \cite[Theorem 2]{korner80}.

\begin{theorem}[Extension of {\cite[Theorem 1, part (a)]{korner80}}]
Eqs.\ (\ref{eq:univs})--(\ref{eq:univi2}) are 
a universally attainable quadruple of exponents and equivocation rate
in the sense of Definition \ref{def:univexp}.
\end{theorem}
\noindent\emph{Proof.}
We shall attach the inverse of two-universal hash functions to
the constant composition code used by K\"orner and Sgarro.
We do not evaluate the decoding error probability,
because that of our code is not larger than \cite{korner80}.
Observe that our exponents in
Eqs.\ (\ref{eq:univs}) and (\ref{eq:univc}) are the same
as \cite{korner80} with the channel $W \circ \Xi$.
We shall evaluate the mutual information.

We assume for a while that $Q_{UV}$ is a type of a sequence of length
$n$ over $\mathcal{U} \times \mathcal{V}$.
Recall that their codebook \cite[Appendix]{korner80} in the random coding
is chosen according to the uniform distribution on the sequences
with joint type $Q_{UV}$.
Let $n$ be the code length,
$\mathcal{B}_n$ the set of private messages for the BCD encoder,
and $\mathcal{E}_n$ the set of common messages.
For $b \in \mathcal{B}_n$
and $e \in \mathcal{E}_n$,
$\lambda(b,e) \in \mathcal{V}^n$ denotes the codeword of $(b,e)$
encoded by the BCD encoder $\lambda$.
$\Lambda$ denotes the random selection of $\lambda$ in the
random coding argument.

Let $\mathcal{S}_n$ be the set of secret messages in the BCC.
Let $\mathcal{F}_n$ be a family of two-universal hash functions
from $\mathcal{B}_n$ to $\mathcal{S}_n$.
For every $f \in \mathcal{F}_n$, we assume that $f$ is surjective
and that $f^{-1}(s)$ has the constant number of elements
for every $s \in \mathcal{S}_n$.
Those assumptions are met, for example, by choosing
the set of all surjective linear maps from $\mathcal{B}_n$
to $\mathcal{S}_n$ as $\mathcal{F}_n$.

The structure of the transmitter and the receiver is as follows:
Fix a hash function $f_n \in \mathcal{F}_n$ and
Alice and Bob agree on the choice of $f_n$.
Given a secret message $s_n$,
choose $b_n$ uniformly randomly from $\{ b \in \mathcal{B}_n \mid
f_n(b) = s_n \}$, treat $b_n$ as the private message to Bob,
encode $b_n$ along with the common message $e_n$ by
a BCD encoder, and get a codeword $v^n$.
Apply the artificial noise to $v^n$ according to the
conditional probability distribution $\Xi$ and
get the transmitted signal $x^n$.
Bob decodes the received signal and get $b_n$, then
apply $f_n$ to $b_n$ to get $s_n$.
This construction requires Alice and Bob to agree on the
choice of $f_n$.

Let $S_n$ denote the RV of the secret message.
Define $B_n$ to be the RV uniformly chosen from the
random set $\{ b \in \mathcal{B}_n \mid
F_n(b) = S_n \}$.
In the following discussion, 
since we treat the channel $W_{\mathcal{Z}}^n\circ \Xi^n:{\cal V}^n\to {\cal Z}^n$,
we simplify it to 
$\overline{W}_{\mathcal{Z}}^n$.
In this case, the mutual information between $F_n(B_n)$ and $Z^n$
depends on the channel from ${\cal V}^n$ to ${\cal Z}^n$.
In particular, for the later analysis, we need to treat the mutual information between $F_n(B_n)$ and $Z^n$
when the channel from ${\cal V}^n$ to ${\cal Z}^n$ is not necessarily memoryless.
So, the mutual information between $F_n(B_n)$ and $Z^n$ will be written 
as a function $I_{\overline{W}_{n}}(F_n(B_n);Z^n|F_n)$ of a discrete channel $\overline{W}_{n}$ from ${\cal V}^n$ to ${\cal Z}^n$.
Note that $\overline{W}_{n}$ is arbitrary and is not necessarily memoryless.

We want to apply the privacy amplification theorem
in order to evaluate 
$I_{\overline{W}_{n}}(F_n(B_n);Z^n|F_n)$.
To use the theorem we must ensure
independence of $F_n$ and $B_n$.
The independence is satisfied by the assumptions on
$\mathcal{F}_n$ if $S_n$ is uniformly distributed.
In that case $B_n$ is uniformly distributed over
$\mathcal{B}_n$.
The remaining task is to find an upper bound
on $I_{\overline{W}_n}(F_n(B_n);Z^n|F_n,\Lambda)$.
%Since the decoding error probability of the above scheme
%is not greater than that of the code for BCD,
%we do not have to analyze the decoding error probability.

Firstly, we consider $\mathbf{E}_{f_n} \exp(\rho I_{W_n}(F_n(B_n);Z^n|F_n=f_n$, $\Lambda=\lambda))$ with
fixed selection $\lambda$ of $\Lambda$.
In the following analysis,
we do not make any assumption on the probability
distribution of $E_n$ except that $S_{n}$, $E_n$, $F_n$ and $\Lambda$
are statistically independent.

Recall that  $\Lambda$ is the RV indicating selection of
codebook in the random ensemble constructed from the joint type $Q_{UV}$
in the way
considered in \cite[Appendix]{korner80}.
Let $U^n = \Lambda(E_n)$ on $\mathcal{U}^n$
and $V^n=\Lambda(B_n,E_n)$ on $\mathcal{V}^n$ codewords
for the BCD taking the random selection $\Lambda$ taking
into account, and $Z^n$ Eve's received signal.
Since we are using the constant composition code as used in
\cite{korner80}, $U^n$ and $V^n$ are not i.i.d.\ RVs.
So, the distribution $P_{V^n,U^n}$ satisfies 
\begin{equation}
P_{V^n|U^n=u}(v) \leq (n+1)^{|\mathcal{U}\times \mathcal{V}|} Q^n_{V|U}(v|u) \Label{eq101}
\end{equation}
for a fixed $u \in \mathcal{U}^n$ 
by \cite[Lemma 2.5, Chapter 1]{csiszarbook}, and
\begin{equation}
P_{U^n}(u) \leq (n+1)^{|\mathcal{U}|} Q^n_U(u), \Label{eq102}
\end{equation}
by \cite[Lemma 2.3, Chapter 1]{csiszarbook}.
Hence, 
(\ref{Haya-2}) yields that
\begin{align}
& \exp (\phi (\rho, \overline{W}_{\mathcal{Z}}^n,P_{V^n|U^n},P_{U^n})) \nonumber \\
\le &
(n+1)^{|\mathcal{U}|^2|\mathcal{V}|}
 \exp (\phi (\rho, \overline{W}_{\mathcal{Z}}^n,Q_{V|U}^n,Q_{U}^n)) 
\label{Haya-30}.
\end{align}

In the code $\Lambda$,
the random variable $V^n$ takes values in 
the subset $T_n(Q_V)$, 
which is defined as the set of elements of ${\cal V}^n$ whose type is $Q_V$.
Hence, it is sufficient to treat the channel whose input system is the subset $T_n(Q_V)$ of ${\cal V}^n$.
Then, we have the following convex combination:
\begin{align}
\overline{W}_{\mathcal{Z}}^n|_{T_n(Q_V)}
=\sum_{\overline{W}_{n}\in {\cal W}_n(Q_V)}
\lambda (\overline{W}_{n}) \overline{W}_{n},
\end{align}
where
$\lambda (\overline{W}_{n})$ is a positive constant and 
${\cal W}_n(Q_V)$ is the family of 
conditional types from ${\cal V}^n$ to ${\cal Z}^n$,
which is the $V$-shell of a sequence of type $Q_V$.
The joint convexity of the conditional relative entropy yields that
\begin{align}
I_{\overline{W}_{\mathcal{Z}}^n}(F_n(B_n);Z^n|F_n)
\le
\sum_{\overline{W}_{n}\in {\cal W}_n(Q_V)}
\lambda (\overline{W}_{n}) 
I_{\overline{W}_{n}}(F_n(B_n);Z^n|F_n).\Label{Haya-21}
\end{align}
We can also show that for any element $\overline{W}_{n}\in {\cal W}_n(Q_V)$ 
we have
%\begin{align}
%e^{\phi(\rho,\overline{W}_{\mathcal{Z}}^n,P_{V^n|U^n},P_{U^n} )}
%\ge
%\lambda (\overline{W}_{n}) 
%e^{\phi(\rho,\overline{W}_{n},P_{V^n|U^n},P_{U^n} )}.
%\end{align}
\begin{align}
&e^{\phi(\rho,\overline{W}_{\mathcal{Z}}^n,P_{V^n|U^n},P_{U^n} )}
\nonumber\\
=&
\sum_{u}P_{U^n}(u)\sum_{z}
(\sum_{v}P_{V^n|U^n}(v|u)
(\sum_{\overline{W}_{n}'\in {\cal W}_n(Q_V)}
\lambda (\overline{W}_{n}') 
\overline{W}_{n}'(z|v))^{\frac{1}{1-\rho}})^{1-\rho} \nonumber\\
\ge &
\sum_{u}P_{U^n}(u)\sum_{z}
(\sum_{v}P_{V^n|U^n}(v|u)
(
%\sum_{\overline{W}_{n}'\in {\cal W}_n(Q_V)}
\lambda (\overline{W}_{n}) 
\overline{W}_{n}(z|v))^{\frac{1}{1-\rho}})^{1-\rho} \nonumber\\
= &
%\sum_{\overline{W}_{n}\in {\cal W}_n(Q_V)}
\lambda (\overline{W}_{n}) 
e^{\phi(\rho,\overline{W}_{n},P_{V^n|U^n},P_{U^n} )}.
\Label{Haya-22}
\end{align}

Hence, in order to evaluate 
$I_{\overline{W}_{\mathcal{Z}}^n}(F_n(B_n);Z^n,E_n|F_n)$,
we evaluate $I_{\overline{W}_{n}}(F_n(B_n);Z^n,E_n|F_n)$:

%By the almost same argument as \cite{matsumotohayashi2011ieice}
%with use of Eq.\ (\ref{hpa1}),
%we can see
{\allowdisplaybreaks
\begin{align}
%&\mathbf{E}_{f_n}\exp(\rho I_{\overline{W}_{n}}(\alpha_{\mathcal{I}}(F_n(B_n));Z^n|F_n=f_n,\Lambda=\lambda))\nonumber\\
&\mathbf{E}_{f_n} \exp(\rho I_{\overline{W}_{n}}((F_n(B_n);Z^n,E_n|F_n=f_n,\Lambda=\lambda))\nonumber\\
%&\textrm{(Giving the common message $E_n$ does not increase $I_{W_n}$ much.)}\nonumber\\
=& \mathbf{E}_{f_n} \exp(\rho  \sum_{e}P_{E_n}(e)
I_{\overline{W}_{n}}(F_n(B_n);Z^n|F_n=f_n,E_n=e,\Lambda=\lambda))\nonumber\\
\leq &
\mathbf{E}_{f_n} \sum_{e}P_{E_n}(e) \exp(\rho  
I_{\overline{W}_{n}}(F_n(B_n);Z^n|F_n=f_n,E_n=e,\Lambda=\lambda))\nonumber\\
\leq &1+ \sum_{e}P_{E_n}(e)
e^{n\rho (R_{\mathrm{s}}-R_{\mathrm{p}})}
\sum_{b,z}P_{B_n}(b)P_{Z^n|B_n,E_n,\Lambda=\lambda}(z|b,e)^{1+\rho}\nonumber\\
&P_{Z^n|E_n=e,\Lambda=\lambda}(z)^{-\rho}\textrm{ (by Eqs.\ (\ref{hpa1}--\ref{hpa1discrete}))}\nonumber\\
=&1+\sum_{e}P_{E_n}(e)
e^{n\rho (R_{\mathrm{s}}-R_{\mathrm{p}})}
\sum_{v,z}\underbrace{\sum_{b:\lambda(b,e)=v}P_{B_n}(b)}_{=P_{V^n|E_n=e,\Lambda=\lambda}(v)}\nonumber\\
&\underbrace{P_{Z^n|B_n,E_n,\Lambda=\lambda}(z|b,e)^{1+\rho}}_{=P_{Z^n|V^n,\Lambda=\lambda}(z|v)^{1+\rho}}P_{Z^n|E_n=e,\Lambda=\lambda}(z)^{-\rho}\nonumber\\
=&1+\sum_{e}P_{E_n}(e)
e^{n\rho (R_{\mathrm{s}}-R_{\mathrm{p}})}
\sum_{v,z}P_{V^n|E_n=e,\Lambda=\lambda}(v)\nonumber\\
&P_{Z^n|V^n,\Lambda=\lambda}(z|v)^{1+\rho}
P_{Z^n|E_n=e,\Lambda=\lambda}(z)^{-\rho}\nonumber\\
=&1+\sum_{e}P_{E_n}(e)
\exp (n\rho (R_{\mathrm{s}} -R_{\mathrm{p}}) + \psi(\rho,P_{Z^n|V^n,\Lambda=\lambda},P_{V^n|E_n=e,\Lambda=\lambda} )
\nonumber\\
=&1+\sum_{e}P_{E_n}(e)
\exp (n \rho (R_{\mathrm{s}} -R_{\mathrm{p}}) + \psi(\rho,P_{Z^n|V^n},P_{V^n|E_n=e,\Lambda=\lambda} )
\nonumber\\
=&1+\sum_{e}P_{E_n}(e)\exp(n\rho(R_{\mathrm{s}}-R_{\mathrm{p}}) 
+ \psi(\rho,\overline{W}_{n}, P_{V^n|E_n=e,\Lambda=\lambda}))\nonumber\\
%&\textrm{ (by \cite{matsumotohayashi2011ieice} and Eq.\ (\ref{eq:psid}))},\nonumber\\
\le &1+\sum_{e}P_{E_n}(e)\exp(n\rho(R_{\mathrm{s}}-R_{\mathrm{p}}) 
+ \phi(\rho,\overline{W}_{n},P_{V^n|E_n=e,\Lambda=\lambda}))\nonumber\\
&\textrm{ (by Eq.\ (\ref{eq:philarger}))}.\nonumber
\end{align}
}

We shall average the above upper bound over $\Lambda$.
By the almost same argument as \cite{matsumotohayashi2011ieice},
we can see
\begin{align}
&\exp(\rho \mathbf{E}_{f_n,\lambda} 
I_{\overline{W}_{n}}(F_n(B_n);Z^n,E_n|F_n=f_n,\Lambda=\lambda)) %\Label{eq:finalprepre}
\nonumber \\
=&\exp(\rho \mathbf{E}_{f_n,\lambda} \sum_{e}P_{E_n}(e)
I_{\overline{W}_{n}}(F_n(B_n);Z^n|F_n=f_n,\Lambda=\lambda,E_n=e)) %\Label{eq:finalprepre}
\nonumber \\
\leq &\mathbf{E}_{f_n,\lambda}\exp(\rho \sum_{e}P_{E_n}(e)
I_{\overline{W}_{n}}(F_n(B_n);Z^n|F_n=f_n,\Lambda=\lambda,E_n=e))\nonumber\\
\leq &1+ \mathbf{E}_{\lambda}\sum_{e}P_{E_n}(e)\exp(n\rho(R_{\mathrm{s}}-R_{\mathrm{p}}) 
+ \phi(\rho,\overline{W}_{n}, P_{V^n|E_n=e,\Lambda=\lambda}))\nonumber\\
\leq &1+ \exp(n\rho(R_{\mathrm{s}}-R_{\mathrm{p}})) \sum_{u\in\mathcal{U}^n} P_{U^n}(u)
\exp(\phi(\rho,\overline{W}_{n},P_{V^n|U^n=u})) \nonumber\\
= &1+ 
\varepsilon_{n,\rho}(\overline{W}_{n}, P_{V^n,U^n})
\Label{eq100},
\end{align}
where
$\varepsilon_{n,\rho}(\overline{W}_{n}, P_{V^n,U^n}):=
\exp(n\rho(R_{\mathrm{s}}-R_{\mathrm{p}})
+\phi(\rho,\overline{W}_{n},P_{V^n|U^n},P_{U^n})))$.
Taking the logarithm, we have
\begin{align}
& \mathbf{E}_{f_n,\lambda} 
I_{\overline{W}_{n}}(F_n(B_n);Z^n,E_n|F_n=f_n,\Lambda=\lambda)
\\
\leq &
\frac{1}{\rho} \log (1+ \varepsilon_{n,\rho}(\overline{W}_{n}, P_{V^n,U^n}) )
\Label{eq100-1},
\end{align}

Observe that what we have shown is that
the averages over $f_n$ and $\lambda$ of 
$\exp(
\rho I_{\overline{W}_{n}}(F_n(B_n);Z^n,E_n|F_n=f_n,\Lambda=\lambda))$
and
$I_{\overline{W}_{n}}(F_n(B_n);Z^n,E_n|F_n=f_n,\Lambda=\lambda))$
are smaller than 
Eqs.\ (\ref{eq100}) and \ (\ref{eq100-1}).

Let $p(n)$ be a polynomial function of $n$.
We can see that with probability of $1-1/p(n)$
the pair $(f_n$, $\lambda)$ makes
$\exp(\rho I_{\overline{W}_{n}}(F_n(B_n);Z^n,E_n|F_n=f_n,\Lambda=\lambda))$
and
$I_{\overline{W}_{n}}(F_n(B_n);Z^n,E_n|F_n=f_n,\Lambda=\lambda)$
smaller than  $p(n)$ times 
Eqs.\ (\ref{eq100}) and (\ref{eq100-1}), respectively.
Since 
the inequalities (\ref{Haya-51}) and (\ref{Haya-52}) hold at least with probability $1-\frac{3}{16}=\frac{13}{16}$
with random selection of $\Lambda$ \cite[Eq. (24)]{korner80},
one can take $p(n) > 2 \frac{16}{13}|{\cal W}_n(Q_V)|$ \cite{csiszarbook}, and
by doing so we can see that there exists at least one pair of
$f_n$ and $\lambda$ such that
all elements $\overline{W}_{n} \in {\cal W}_n(Q_V)$ satisfies
the inequalities (\ref{Haya-51}) and (\ref{Haya-52}) and
\begin{align}
& I_{\overline{W}_{n}}(F_n(B_n);Z^n,E_n |F_n=f_n,\Lambda=\lambda)
\nonumber \\
\le & 
\frac{p(n)}{\rho}  \log (1+ \varepsilon_{n,\rho}(\overline{W}_{n}, P_{V^n,U^n}))
\le 
\frac{p(n)}{\rho} \varepsilon_{n,\rho}(\overline{W}_{n}, P_{V^n,U^n})
 \Label{Haya-13} \\
& \exp(\rho I_{\overline{W}_{n}}(F_n(B_n);Z^n,E_n|F_n=f_n,\Lambda=\lambda))
\nonumber \\
\le & p(n) (1+ \varepsilon_{n,\rho}(\overline{W}_{n}, P_{V^n,U^n}) ). \Label{Haya-4}
\end{align}
Thus, 
(\ref{Haya-21}), (\ref{Haya-22}), (\ref{Haya-30}) and (\ref{Haya-13}) yield that
\begin{align}
&
I_{\overline{W}_{\mathcal{Z}}^n}(F_n(B_n);Z^n,E_n|F_n=f_n,\Lambda=\lambda) \nonumber\\
%=& \sum_{e}P_{E_n}(e)I_{\overline{W}_{\mathcal{Z}}^n}(F_n(B_n);Z^n|F_n=f_n,\Lambda=\lambda,E_n=e ) 
%\Label{Haya-23-1}
%\nonumber \\
%\le &
%\sum_{\overline{W}_{n}\in {\cal W}_n(Q_V)}
%\lambda (\overline{W}_{n}) 
%\frac{p(n)}{\rho}  
%\varepsilon_{n,\rho}(\overline{W}_{n}, P_{V^n,U^n})
%\nonumber \\
\le &
\sum_{\overline{W}_{n}\in {\cal W}_n(Q_V)}
\lambda (\overline{W}_{n}) 
\frac{p(n)}{\rho} 
\varepsilon_{n,\rho}(\overline{W}_{n}, P_{V^n,U^n})
\nonumber \\
\le &
\sum_{\overline{W}_{n}\in {\cal W}_n(Q_V)}
\frac{p(n)(n+1)^{|\mathcal{U}|^2|\mathcal{V}|}}{\rho}  \varepsilon_{n,\rho}(\overline{W}_{\mathcal{Z}}^n,Q_{V,U}^n) 
\nonumber \\
\le &
\frac{p(n)|{\cal W}_n(Q_V)|(n+1)^{|\mathcal{U}|^2|\mathcal{V}|}}{\rho}  \varepsilon_{n,\rho}(\overline{W}_{\mathcal{Z}}^n,Q_{V,U}^n) 
\nonumber \\
=&
\frac{\overline{p}(n)}{\rho}  
\varepsilon_{1,\rho}(\overline{W}_{\mathcal{Z}},Q_{V,U})^n 
%=&
%\frac{\overline{p}(n)^2}{\rho}  
%\exp (n (\rho 
%(R_{\mathrm{p}}-R_{\mathrm{s}})+\phi (\rho,\overline{W}_{\mathcal{Z}},Q_{V|U},Q_U)))
\Label{Haya-23},
\end{align}
where
$\overline{p}(n):=p(n)(n+1)^{|\mathcal{U}|^2|\mathcal{V}|}|{\cal W}_n(Q_V)|$.

Since $
\log \varepsilon_{1,\rho}(\overline{W}_{\mathcal{Z}},Q_{V,U})
=
R_{\mathrm{s}}-R_{\mathrm{p}}+\phi (\rho,\overline{W}_{\mathcal{Z}},Q_{V|U},Q_U)$,
for an arbitrary $\delta>0$,
we can choose a large integer $n_1$ such that
\begin{align}
& \inf_{1/n \le \rho \le 1}
\log \bigl(\frac{\overline{p}(n)}{\rho}  
\varepsilon_{1,\rho}(\overline{W}_{\mathcal{Z}},Q_{V,U})^n \bigr) \nonumber \\
\le &
\inf_{1/n \le \rho \le 1}
\log 
\bigl( \varepsilon_{1,\rho}(\overline{W}_{\mathcal{Z}},Q_{V,U})^n \bigr) 
+\log \overline{p}(n)+\log n
\nonumber \\
\le &
-n ( F_+^{\mathcal{I}}(W, R_{\mathrm{p}}, R_{\mathrm{s}}, Q_{UV},\Xi ) -\delta)
\nonumber 
\end{align}
for $n \ge n_1$.
Since
(\ref{Haya-23}) holds with any $\rho \in [1/n,1]$, 
we obtain
\begin{align}
& \log 
I_{\overline{W}_{\mathcal{Z}}^n}(F_n(B_n);Z^n,E_n|F_n=f_n,\Lambda=\lambda) \nonumber \\
\le &
-n ( F_+^{\mathcal{I}}(W, R_{\mathrm{p}}, R_{\mathrm{s}}, Q_{UV},\Xi ) -\delta)
\Label{Haya-03}
\end{align}
for $n \ge n_1$.

Since $x \mapsto \exp(x)$ is convex,
(\ref{Haya-21}), (\ref{Haya-22}), (\ref{Haya-30}) and (\ref{Haya-4}) yield that
\begin{align}
& \exp(\rho I_{\overline{W}_{\mathcal{Z}}^n}(F_n(B_n);Z^n,E_n|F_n=f_n,\Lambda=\lambda )) \nonumber \\
\le &
\sum_{\overline{W}_{n}\in {\cal W}_n(Q_V)}
\lambda (\overline{W}_{n}) 
 \exp(\rho I_{\overline{W}_n}(F_n(B_n);Z^n,E_n|F_n=f_n,\Lambda=\lambda )) 
\nonumber \\
\le &
\sum_{\overline{W}_{n}\in {\cal W}_n(Q_V)}
\lambda (\overline{W}_{n}) 
p(n)(1+  \varepsilon_{n,\rho}(\overline{W}_{n},P_{V^n,U^n}))
\nonumber \\
\le &
\sum_{\overline{W}_{n}\in {\cal W}_n(Q_V)}
p(n)(1+  \varepsilon_{n,\rho}(\overline{W}_{\mathcal{Z}}^n,P_{V^n,U^n})) \nonumber \\
\le &
p(n)|{\cal W}_n(Q_V)|(1+  \varepsilon_{n,\rho}(\overline{W}_{\mathcal{Z}}^n,P_{V^n,U^n}))\nonumber\\
\le &
\overline{p}(n)(1+  
\varepsilon_{n,\rho}(\overline{W}_{\mathcal{Z}}^n,Q_{V,U}^n))
\nonumber .
\end{align}

Taking the logarithm,
we have
\begin{align}
& 
I_{\overline{W}_{\mathcal{Z}}^n}(F_n(B_n);Z^n,E_n|F_n=f_n,\Lambda=\lambda )
\nonumber \\
%\Label{Haya-50} \\
\le &
\frac{
\log \overline{p}(n)^2(1+  
\varepsilon_{n,\rho}(\overline{W}_{\mathcal{Z}}^n,Q_{V,U}^n))
}{\rho}
\nonumber \\
\le &
\frac{\log (2 \overline{p}(n)^2)}{\rho}
+
n 
\frac{[\log \varepsilon_{1,\rho}(\overline{W}_{\mathcal{Z}},Q_{V,U})))]_+ }{\rho} 
%\sqrt{n}\log (2 \overline{p}(n)^2)
% + \sqrt{n}\varepsilon_{n,1/\sqrt{n}}(\overline{W}_{\mathcal{Z}}^n,Q_{V,U}^n))) 
\Label{Haya-24}.
\end{align}
Since 
$\lim_{\rho \to 0}
\frac{[\log \varepsilon_{1,\rho}(\overline{W}_{\mathcal{Z}},Q_{V,U})))]_+ }{\rho} 
=
F_-^{\mathcal{I}}(W, R_{\mathrm{p}}, R_{\mathrm{s}}, Q_{UV},\Xi )$,
we can choose an integer $n_2$ such that
\begin{align}
& 
I_{\overline{W}_{\mathcal{Z}}^n}(F_n(B_n);Z^n,E_n |F_n=f_n,\Lambda=\lambda)
% \textrm{Eq.} (\ref{Haya-50}) 
\nonumber \\
\le &
\textrm{RHS of} (\ref{Haya-24}) \hbox{ with }  \rho=1/\sqrt{n} \nonumber \\
\le &
n (F_-^{\mathcal{I}}(W, R_{\mathrm{p}}, R_{\mathrm{s}}, Q_{UV},\Xi )+\delta ) \Label{Haya-02}
\end{align}
for $n \ge n_2$

Therefore, using (\ref{Haya-03}), (\ref{Haya-02}),
we can see that $(F^{\mathrm{s}}$, $F^{\mathrm{c}}$, 
$F_+^{\mathcal{I}})$
$F_-^{\mathcal{I}})$
is a universally attainable quadruple of exponents
in the sense of Definition \ref{def:univexp}.
\qed

\begin{remark}
By suitably changing $R_{\mathrm{p}}$, $Q_{UV}$ and $\Xi$ in
Eqs.\ (\ref{eq:univs})--(\ref{eq:univi2}),
the coding scheme used in the proof can achieve
a rate triple $(R_{\mathrm{s}}$, $R_e$, $R_{\mathrm{c}})$
if there exists a Markov chain $U\rightarrow V \rightarrow X \rightarrow 
YZ$ and
\begin{eqnarray}
R_{\mathrm{s}} &\leq&I(V;Y|U), \label{eq:extra}\\
%R_{\mathrm{s}} + R_{\mathrm{c}} &\leq& I(V;Y|U)+\min[I(U;Y),I(U;Z)],\nonumber\\
R_{\mathrm{c}} &\leq& \min[I(U;Y),I(U;Z)],\nonumber\\
R_e & \leq & I(V;Y|U)-I(V;Z|U),\nonumber\\
R_e & \leq &R_{\mathrm{s}}.\nonumber
\end{eqnarray}
Observe that Eq.\ (\ref{eq:extra}) does not exist in 
Theorem \ref{th1}, and that our achievable region could be smaller.
The reason behind this difference is that 
we do not split the confidential message into the private message
$B_n$ and the common message $E_n$ encoded by the BCD encoder.
The coding scheme for BCC in \cite{csiszar78} uses this kind
of message splitting. 

However, when $R_e = R_{\mathrm{s}}$, our coding scheme can achieve the
rate pairs given in Corollary \ref{cor:bcc} 
by suitably changing $R_{\mathrm{p}}$, $Q_{UV}$ and $\Xi$ in
Eqs.\ (\ref{eq:univs})--(\ref{eq:univi2}),
because the message splitting is unnecessary when $R_e = R_{\mathrm{s}}$.
\end{remark}

\section{Conclusion}\Label{sec4}
In this paper, we presented universally attainable error and
information exponents 
universally attainable equivocation rates
for discrete broadcast channels with
confidential messages. The result is novel as far as the authors know.
However, there are still rooms for improving this research result.
K\"orner and Sgarro also clarified upper bounds on the error
exponents, but we could not obtain one, because it is difficult to
evaluate the smallest possible mutual information leaked to
Eve over all the possible coding schemes.
On the other hand, the non-universal information exponent
appeared in \cite{hayashi11} is better than one presented
here. This suggests that our universally attainable information
exponent could be improved.

\section*{Acknowledgment}
The first author would like to thank 
Dr.\ Jun Muramatsu and Prof.\ Tomohiro Ogawa for the helpful discussion on the universal
coding.
A part of this research was done during the first author's stay
at the Institute of Network Coding, the Chinese University
of Hong Kong, and he greatly appreciates the hospitality by Prof.\ Raymond
Yeung.
This research was partially supported by 
the MEXT Grant-in-Aid for Young Scientists (A) No.\ 20686026 and
(B) No.\ 22760267, and Grant-in-Aid for Scientific Research (A) No.\ 23246071.
The Center for Quantum Technologies is funded
by the Singapore Ministry of Education and the National Research
Foundation as part of the Research Centres of Excellence programme.

%\bibliographystyle{IEEEtranS}
%\bibliography{mrabbrev,mybib}

% Generated by IEEEtranS.bst, version: 1.12 (2007/01/11)

\end{document}